\documentclass[a4paper,fleqn]{cas-sc}

\usepackage[numbers]{natbib}

\usepackage{amsmath}
\usepackage{graphics}
\usepackage{epsfig}
\usepackage{mathrsfs}
\usepackage{amsmath,amssymb,amsthm}
\usepackage{tikz}
\usepackage{flushend}

\usepackage{multirow}

\usepackage{color}
\usepackage{array}
\newtheorem{lemma}{Lemma}

\newtheorem{corollary}{Corollary}
\usepackage[linesnumbered,ruled,vlined]{algorithm2e}
\usepackage{algpseudocode}
\usepackage{amsmath}

\usepackage{graphicx}

\def\tsc#1{\csdef{#1}{\textsc{\lowercase{#1}}\xspace}}
\tsc{WGM}
\tsc{QE}
\tsc{EP}
\tsc{PMS}
\tsc{BEC}
\tsc{DE}

\begin{document}
\let\WriteBookmarks\relax
\def\floatpagepagefraction{1}
\def\textpagefraction{.001}
\shortauthors{Shanshan Wang et~al.}

\title [mode = title]{Algorithms with improved delay for enumerating connected induced subgraphs of a large cardinality}

\author[1]{Shanshan Wang}

\address[1]{Shantou University, China}

\author[1]{Chenglong Xiao}
\cormark[1]

\author%
[2]
{Emmanuel Casseau}

\address[2]{Univ Rennes, INRIA, CNRS, IRISA, France}

\cortext[cor1]{Corresponding author:Chenglong Xiao, email:chlxiao@stu.edu.cn}

\begin{abstract}
The problem of enumerating all connected induced subgraphs of a given order $k$ from a given graph arises in many practical applications: bioinformatics, information retrieval, processor design,to name a few. The upper bound on the number of connected induced subgraphs of order $k$ is $n\cdot\frac{(e\Delta)^{k}}{(\Delta-1)k}$, where $\Delta$ is the maximum degree in the input graph $G$ and $n$ is the number of vertices in $G$. In this short communication, we first introduce a new neighborhood operator that is the key to design reverse search algorithms for enumerating all connected induced subgraphs of order $k$. Based on the proposed neighborhood operator, three algorithms with delay of $O(k\cdot min\{(n-k),k\Delta\}\cdot(k\log{\Delta}+\log{n}))$, $O(k\cdot min\{(n-k),k\Delta\}\cdot n)$ and $O(k^2\cdot min\{(n-k),k\Delta\}\cdot min\{k,\Delta\})$ respectively are proposed. The first two algorithms require exponential space to improve upon the current best delay bound $O(k^2\Delta)$\cite{4} for this problem in the case $k>\frac{n\log{\Delta}-\log{n}-\Delta+\sqrt{n\log{n}\log{\Delta}}}{\log{\Delta}}$ and $k>\frac{n^2}{n+\Delta}$ respectively.
\end{abstract}


\maketitle

\section{Introduction}

The problem of enumerating all connected induced subgraphs of order $k$ is involved in many applications. Such applications include for example identifying network motifs from biological networks \cite{1}, retrieving keyword queries over RDF graphs \cite{2} and enforcing higher consistency levels in constraint processing \cite{3}. Following is the definition of the problem that this short communication focuses on.

\textit{Problem $GEN(G;k)$ }: Given an undirected graph $G = (V,E)$, the problem is to enumerate all subsets $X\subset V$ of vertices such that $|X|=k$ and the subgraph $G[X]$ induced on $X$ is connected.

 The upper bound on the number of connected induced subgraphs of order $k$ is $n\cdot\frac{(e\Delta)^{k}}{(\Delta-1)k}$, where $\Delta$ is the maximum degree in the input graph $G$ \cite{9,10}. Hence, the problem $GEN(G;k)$ is a computationally difficult problem. An efficient algorithm is of great importance.  In the literature, the algorithms for subgraph enumeration problem are usually evaluated in terms of delay for worst-case running time analysis. The delay is the maximal time that the algorithms spend between two successive outputs. In this work, we also use the upper bound on the delay of the enumeration algorithms as a nontrivial running time bound.

Most previous approaches perform the enumeration procedure by incrementally enlarging the connected subgraphs until the size of such subgraphs is $k$ \cite{1,4,5}. In these approaches, the small subgraphs are expanded by absorbing neighbor vertices. The algorithm presented in \cite{1} starts by assigning each vertex a number as a unique label. Then, the subgraphs are gradually expanded by adding neighbor vertex that has larger label and be neighbored to the newly added vertex but not to a vertex already in the subgraph. In such way, each subgraph is enumerated exactly once. A variant of this algorithm is introduced in the most recent literature \cite{4}. The variant algorithm (denoted as $Simple$) introduces a pruning rule that can avoid unnecessary recursion. With the introduced pruning rule, the algorithm $Simple$ achieves a delay of $O(k^{2}\Delta)$. Furthermore, a so-called $k$-component rule is applied to speed up the enumeration of connected induced subgraphs of large cardinality. Another recent approach introduced in \cite{5} tries to expand the connected induced subgraphs by adding only the validated neighbors. Each neighbor vertex is validated by judging if it has greater distance to the anchor vertex $v$ (the vertex with the smallest vertex identifier in the current subgraph) than the distance from the utmost vertex $u$ (the vertex in the current subgraph with the longest shortest path to the anchor vertex) to $v$ or it has the same distance to $v$ as the distance from $u$ to $v$ but it is lexicographically greater than $u$.

Different from the bottom-up approaches, Elbassioni proposed a reverse search based algorithm for enumerating all connected induced subgraphs of order $k$ in \cite{6}. We call this algorithm $RwD$ ($Reverse Search with Dictionary$) as referred to in \cite{4}. Please note that the idea of reverse search method for enumeration was first introduced in \cite{7}. The reverse search algorithm produces all the solutions by traversing the supergraph $\mathcal{G}$. Each node of $\mathcal{G}$ corresponds to a solution of problem $GEN(G;k)$, that is, a connected vertex set of order $k$. The arcs between nodes $X$ and $X'$ in $\mathcal{G}$ are defined by a neighbor operator: $\mathcal{N}(X)=\{X'\in \mathcal{C}(G;k):|X\cap X'|=k-1\}$, where $\mathcal{C}(G;k)$ denotes the family of vertex sets of order $k$ in $G$. In the following, as in \cite{6}, we distinguish the vertices of $G$ and $\mathcal{G}$ by denoting them respectively as vertices and nodes.

The $RwD$ algorithm initially generates a solution $X_0$ ($G[X_{0}]$ is a connected induced subgraph of order $k$). Based on the initial solution, all the neighbors of $X_0$ in $\mathcal{G}$ are visited by exchanging one vertex. In detail, a vertex $u$ in $X_0$ is deleted, and a common neighbor vertex $v$ of connected components in $X_{0}\setminus\{u\}$ is added to form a connected vertex set $X_{0}\setminus\{u\}\cup\{v\}$. In order to avoid duplicates, each newly generated solution should be checked if it has already been visited. Then, each unvisited neighbor is added to a list or a queue as a candidate to be further explored. This procedure can be carried out in depth-first way or breath-first way until all nodes in $\mathcal{G}$ are visited. The pseudo code of the algorithm can be found in \cite{4}. According to the definition of neighborhood presented in \cite{6}, $|\mathcal{N}(X)|\leq k\cdot min\{(n-k),k\Delta\}$, and each solution is generated in time $O(k(\Delta+\log{k})+\log{n}))$. Thus, the delay of the algorithm is $O(k\cdot min\{(n-k),k\Delta\}\cdot(k(\Delta+\log{k})+\log{n}))$.

\section{Proposed Approaches}

\begin{algorithm}[t]\scriptsize
  \caption{The $IRwD$ Algorithm}
  \KwIn{An undirected graph $G=(V,E)$}
  \KwOut{$\mathcal{K}$ A set of enumerated subgraphs}
  $Queue~~Q = \emptyset$\;
  $\mathcal{K}=\emptyset$\;
  \For{each connected component $C$ in $G$}
  {
    $S = an~initial~solution~in~C$\;
    $Q.enqueue(S)$\;
    $\mathcal{K}=\mathcal{K}\cup\{S\}$\;
    \While{$Q\neq\emptyset$}
    {
       $S = Q.dequeue()$\;
       $output~ S$\;
       $A = ArticulationPoint(S)$\;
      \For{each vertex $v \in S\setminus A$}
      {
         $S' = S\setminus \{v\}$\;
          \For{each neighbor vertex $w$ of $S'$}
          {
                $S'' = S'\cup\{w\}$\;
                \If{$S''\notin \mathcal{K}$}
                {
                      $Q.enqueue(S'')$\;
                      $\mathcal{K}=\mathcal{K}\cup\{S''\}$\;
                }
          }

      }
  }
  }
\end{algorithm}

We assume the input undirected graph is a connected graph, if not we can simply deal with each connected component separately. In \cite{6}, two connected sets $X$ and $X'$ of order $k$ are neighbors if they have $k-1$ vertices in common (the induced subgraph of the $k-1$ vertices in common can be disconnected or connected). In this work, we introduce a different definition of neighborhood of $\mathcal{G}$. For a set $X\in \mathcal{C}(G;k)$, the neighbors of $X$ are obtained from $X$ by exchanging one vertex: $\mathcal{N}(X)=\{X'\in \mathcal{C}(G;k):|X\cap X'|=k-1~and~G[X\cap X']~is~connected\}$ . In other words, any pair of two nodes in $\mathcal{G}$ are neighbors only if the two nodes have $k-1$ vertices in common and the induced graph of their intersection is also connected.  An illustrative example of the supergraph method \cite{7} based on the new neighborhood definition and the neighborhood definition proposed in \cite{6} is shown in Fig. 1. If the supergraph $\mathcal{G}$ constructed according to the introduced new definition of neighborhood is strongly connected, we then can explore all the nodes in $\mathcal{G}$ more efficiently starting from any node in it.

\begin{figure}[t]
  \centering
  \includegraphics[width=0.52\hsize=0.8]{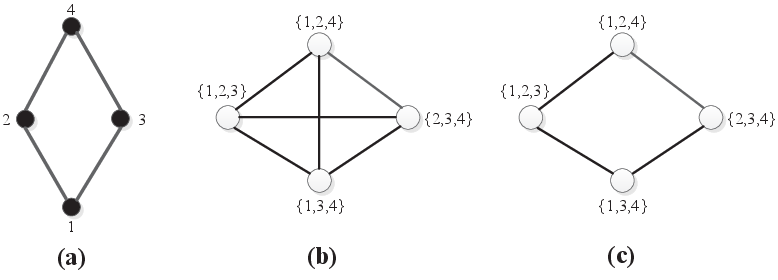}
  \caption{An illustrative example of constructing the supergraph $\mathcal{G}$ ($k=3$). (a) a simple graph $G$ with four vertices. (b) the supergraph  $\mathcal{G}$ constructed based on the neighborhood operator introduced in \cite{6}. (c) the supergraph  $\mathcal{G}$ constructed based on the proposed new neighborhood operator.}
  \label{f2}
\end{figure}

~\newline
Now, we prove that the supergraph $\mathcal{G}$ is strongly connected.

\begin{lemma}
  Let $X$,$Y$ be two distinct elements in $\mathcal{C}(G;k)$. Then there exists vertex sets $X_1$,$X_2$,$\cdot\cdot\cdot$,$X_l$ $\in \mathcal{C}(G;k)$ such that $X_1=X$,$X_l=Y$, $l\leq n-k+1$, and for $i=1,\cdot\cdot\cdot,l-1$, $X_{i+1}\in\mathcal{N}(X_i)$.
\end{lemma}

\begin{proof}
  Similarly to the proof of Lemma 1 in \cite{6}, we first define the $d(Z,Z')$ as the shortest distance between the two vertex sets $Z$,$Z'$ in $G$ ($Z,Z'\in \mathcal{C}(G;K)$). We consider two cases: $d(X,Y)=0$ and $d(X,Y)>0$.

  Case 1. If $d(X,Y)=0$, then there exists at least a vertex $z \in X\cap Y$. We first contract the connected component containing $z$ in $X\cap Y$ as a single vertex $v$ in $G[X]$ and denote the new graph by $G'$. As $G'$ is connected, it has a spanning tree $T$. $T$ has a leaf $u\neq v$. Furthermore, there exists at least a vertex $w\in Y$, $w\notin X$ and $w$ is the neighbor of $v$. We delete $u$ from $X$ and add $w$ to $X$, we then have $X_2=X\cup\{w\}\setminus\{u\}$. We can iteratively perform this procedure to construct $X_3,...X_i$ until $X_i=Y$, where $i\leq k-1$\footnote{There are at most $k-1$ vertices that are in $X$ and not in $Y$. Thus, it is required to remove at most $k-1$ vertices from $X$ to get $Y$.}.

  Case 2. If $d(X,Y)>0$, then $X\cap Y=\emptyset$. Let $v_0,v_1,\cdot\cdot\cdot,v_j$ be the ordered sequence of vertices on the shortest path between $X$ and $Y$ in $G$, where $v_0\in X$ and $v_j\in Y$. Let $T$ be a spanning tree of $G[X]$. Then $T$ has at least a leaf $u\neq v_0$. We delete $u$ and add $v_1$, then we have $X_{2}=X_1\cup\{v_1\}\setminus\{u\}$ ($G[X_{2}]$ is a connected subgraph). We can continue this procedure to construct $X_3,...,X_i$ until $d(X_i,Y)=0$. As $d(X,Y)\leq n-2k$, we will arrive at case 1 after at most $n-2k+1$ such iterations.
\end{proof}

Since the supergraph $\mathcal{G}$ is strongly connected, we can traverse all the solutions (nodes) in $\mathcal{G}$ starting from any node in $\mathcal{G}$. The pseudo code of our proposed algorithm is shown in Algorithm 1. This algorithm, called $IRwD$ ($Improved Reverse Search with Dictionary$) differs from $RwD$ \cite{6} in two points. First, as the induced subgraph of the common part of two neighbors should be connected, only the vertex that is not an articulation point of $S$ can be deleted (the articulation points of $S$ can be quickly found in $O(k\cdot\min\{k,\Delta\})$ time by calling the algorithm proposed by Tarjan\cite{8}, i.e., line 10 of Algorithm 1). Second, as the approach of \cite{6} deletes a vertex in $S$ in each iteration, $S'$ may not be connected. Hence, finding common neighbors of connected components of $S'$ is required to ensure the connectivity of the generated solution, and it is one of most time-consuming procedure inside each call. The searching of the common neighborhood of connected components of $S'$ takes $O(k(\Delta+\log{k}))$ time \cite{6}. This time-consuming job can be avoided in our approach.

~\newline
The delay of $IRwD$ algorithm is given and proven as follows.

\begin{lemma}
   The delay of the proposed algorithm $IRwD$ is $O(k\cdot min\{(n-k),k\Delta\}\cdot(k\log{\Delta}+\log{n}))$, where $\Delta$ is the maximum degree of $G$.
\end{lemma}

\begin{proof}

It is clear that the first solution can be found in $O(k\Delta)$ by traversing $G$ in depth-first search (DFS) starting from an arbitrary vertex. Now, we show that in $O(k\cdot min\{(n-k),k\Delta\}\cdot(k\log{\Delta}+\log{n}))$ time we either find a new node (subgraph) or terminate the enumeration. We compute the set of all the articulation points ($A$) in $S$ by calling Tarjan's algorithm before the nested for loop (line 10, Algorithm 1). The running time of Tarjan's algorithm is $O(k\cdot\min\{k,\Delta\})$. Given a node $S$, the nested for loop (lines 11-17, Algorithm 1) requires at most $N(T1+T2)$ time, where $N$ is the maximum number of the neighbors of $S$, $T1$ is the time to generate a neighbor and $T2$ is the time to check if the neighbor has already been generated before. Based on the introduced new definition of neighborhood, we have $N \leq k\cdot min\{(n-k),k\Delta\}$, and by the introduced neighborhood operator, $T1=O(1)$. By maintaining a priority queue $\mathcal{K}$ on the set of generated solutions (the subgraphs that are already discovered and processed), we can ensure that $T2=\log{|S|}$ using a balanced binary search tree on the solutions generated, where $|S|$ is the maximum number of generated subgraphs. The currently known bound on $|S|$ is $O(n\cdot \frac{(e\Delta)^k}{(\Delta-1)k})$, which implies that $T2$ is $O(k\log\Delta+\log n)$. Thus, the time for executing the nested for loop is at most $O(k\cdot min\{(n-k),k\Delta\}\cdot(k\log{\Delta}+\log{n}))$. After executing the nested for loop, we have two cases: If the queue $Q$ is empty, the algorithm terminates; If the queue $Q$ is not empty, we pick the first node from the queue and output it. Therefore, the overall delay of $IRwD$ is $O(k\cdot min\{(n-k),k\Delta\}\cdot(k\log{\Delta}+\log{n}))$.
\end{proof}

In \cite{4}, an algorithm called $Simple$ with a delay of $O(k^2\Delta)$ is presented. This delay is the current best delay in the literature. Compared with the best delay proven in \cite{4}, it can be seen that, for small values of $k$ the delay of $Simple$ is better than the delay of $IRwD$ and for large $k$ (e.g., $k$ close to $n$) the delay of $IRwD$ is better. More precisely, the $IRwD$ algorithm has a better delay than the $Simple$ algorithm in the case  $k>\frac{n\log{\Delta}-\log{n}-\Delta+\sqrt{n\log{n}\log{\Delta}}}{\log{\Delta}}$. In the proposed $IRwD$ algorithm, similar to $RwD$\cite{6}, we also maintain a balanced binary search tree on the solutions generated. Each time a solution is generated, we check whether the solution has already been stored in the balanced binary search tree. In this work, we propose another method to check whether a solution has been already visited. Instead of using a balanced binary search tree to store the solutions visited, we maintain a binary search tree on the sequences of visited vertices. The height of the binary search tree is $n$, where $n$ is the number of vertices in $G$. In the binary search tree, each node corresponds to a vertex in $G$. 1- and 0-branches at node $i$ represents the addition or not of the vertex $i$ of $G$ respectively. Fig.2 shows an example of the binary search tree constructed for the set of discovered solutions $\{\{1,2,3\},\{1,2,4\},\{2,3,4\}\}$.  With this method, we have a variant algorithm of $IRwD$, and we arrive at the following.

\begin{figure}[t]
  \centering
  \includegraphics[width=0.32\hsize=0.8]{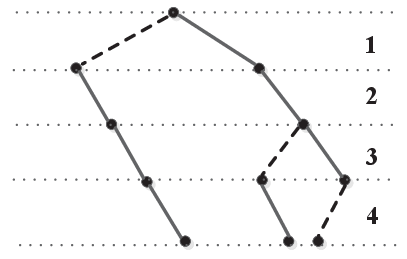}
  \caption{The binary search tree corresponds to the set of discovered solutions $\{\{1,2,3\},\{1,2,4\},\{2,3,4\}\}$ }
  \label{f2}
\end{figure}

\begin{corollary}
The variant algorithm of $IRwD$ solves $GEN(G;k)$ for any graph $G$ with delay $O(k\cdot min\{(n-k),k\Delta\}\cdot n)$, where $\Delta$ is the maximum degree of $G$.
\end{corollary}

\begin{proof}

The only difference between the variant algorithm and $IRwD$ is the time spent on checking if the neighbor has already been generated before. With the binary search tree on the visited sequences of vertices, we can check whether a neighbor has already been visited in $O(n)$ time, which implies the overall delay of the variant algorithm is  $O(k\cdot min\{(n-k),k\Delta\}\cdot n)$. 
\end{proof}

Comparing the aforementioned delay $O(k\cdot min\{(n-k),k\Delta\}\cdot n)$ with the current best delay $O(k^2\Delta)$ in the literature, it can be seen that the variant algorithm achieves a better delay in the case $k>\frac{n^2}{n+\Delta}$. If $k>\frac{n-\log n}{\log \Delta}$, then $n<k\log{\Delta}+\log{n}$. Thus, the check if the neighbor has already been generated using proposed binary search tree is faster than the check using balanced binary tree \cite{6} in this case. 

Furthermore, it should be noted that both the proposed algorithm and the algorithm presented in \cite{6} use a dictionary to store all previously detected solutions. Hence, the algorithms requires exponential space ($O(n+m+k|\mathcal{C}(G,k)|)$) (the variant algorithm of $IRwD$ has a slightly improved space bound of $O(n+m+|\mathcal{C}(G,k)|)$ ). However, the algorithm $Simple$ of \cite{4} requires only linear space. In order to avoid the use of exponential space, the author of \cite{6} also proposed another algorithm called $RwP$ that has slightly worse delay ($O((k\cdot min\{(n-k),k\Delta\})^2\cdot(\Delta+\log{k}))$) but requires only linear space. The $RwP$ algorithm also used the reverse search method. The main difference between $RwD$ and $RwP$ is that $RwP$ applies a parent function to ensure every neighbor is enumerated exactly once instead of storing all the discovered subgraphs. In $RwP$, all solutions are sorted lexicographically and each solution has a unique parent. A modified version of the algorithm is presented in \cite{4}. The modified algorithm has a delay of  $O(k^3\cdot\Delta \cdot min\{(n-k),k\Delta\})$ by using a slightly different predecessor check method.

As the supergraph $\mathcal{G}$ constructed by the new neighbor operator is strongly connected, according to the definition of predecessor function in \cite{4} and the definition of parent node in \cite{6}, the methods of finding the predecessor or the parent can also be used in our proposed algorithm in order to have a linear space bound. The pseudo code of our algorithm that adopts the predecessor check to avoid multiple enumerations can be found in Algorithm 2 (the algorithm is denoted as $IRwP$). The predecessor of a connected induced subgraph $S$ is defined as follows: let $S$ be a connected induced subgraph of order $k$, $S\setminus \{u\}\cup\{v\}$ is the predecessor of $S$ if $S\setminus \{u\}\cup \{v\}$ and $S\setminus\{u\}$ are connected, where the vertex $u\in S$ with lowest index and the vertex $v\notin S$ with highest index. The algorithm works as follows: We first assign each vertex in $G$ a number based on depth-first search rather than an arbitrary order. The depth-first search implies a lexicographical ordering of the solutions. Each solution has a unique predecessor according to the defined lexicographic order. We then start from a lexicographically largest solution (the lexicographically largest solution can be found in time $O(k\Delta)$ by traversing the DFS tree that defines the lexicographic order on the vertex sets.), and traverse the neighbors of each node $S$ in $\mathcal{G}$ with DFS or breath-first search  (BFS). If a neighbor $S'$ of $S$ is considered, we only output $S'$ (or put it into the queue) if $S$ is the predecessor of $S'$.

The original algorithm proposed in \cite{6} adopts DFS to find all the connected induced subgraphs of order $k$. Komusiewicz, et al. implemented the algorithm ($RwP$) with BFS. It worthy noted that the implementation with BFS has a space bound of $O(n+m+k|\mathcal{C}(G,k)|)$. However, the claimed delay and a linear space can be ensured if we implement it with DFS and distinguish between nodes of odd and even depth in the search tree \cite{4,6}. For a straightforward comparison with the most recent work in \cite{4}, we also provide the BFS version of $IRwP$ in this letter.

\begin{algorithm}[t]\scriptsize
  \caption{The $IRwP$ Algorithm}
  \KwIn{An undirected graph $G=(V,E)$}
  \KwOut{A set of enumerated subgraphs}
  $Queue~~Q = \emptyset$\;
  \For{each connected component $C$ in $G$}
  {
    $S = lexicographically~~largest~solution~in~C$\;
    $Q.enqueue(S)$\;
    \While{$Q\neq\emptyset$}
    {
       $S = Q.dequeue()$\;
       $output~ S$\;
       $A = ArticulationPoint(S)$\;
      \For{each vertex $v \in S\setminus A$}
      {
          $S' = S\setminus \{v\}$\;
          $N = neighbor~vertices~of~S'$\;
          \For{ vertex $w \in N$}
          {
                $S'' = S'\cup\{w\}$\;
                $A' = ArticulationPoint(S'')$\;
                \If{ $w$ is the vertex with lowest index in $S''\setminus A'$ and $v$ is the vertex with highest index in $N$}
                {
                        $Q.enqueue(S'')$\;
                }
          }

      }
  }
  }
\end{algorithm}

\begin{lemma}
   The delay of the proposed algorithm $IRwP$ is  $O(k^2\cdot min\{(n-k),k\Delta\}\cdot min\{k,\Delta\})$, where $\Delta$ is the maximum degree of $G$.
\end{lemma}

\begin{proof}

We first show the time spent on checking if $S$ is the predecessor of $S''$ (lines 14-16, Algorithm 2). The predecessor check is executed as follows: We find the articulation points of $S''$, this can be done in $O(k\cdot\min\{k,\Delta\})$ time. Then, we check if the vertex $w$ is the vertex lowest index in $S''\setminus A'$ and the vertex $v$ is the vertex with highest index in $N$. This check takes $O(1)$ time \footnote{To achieve the claimed $O(1)$ time, we can pick the vertex with lowest index from $S''$ and $A'$ respectively when we generate $S''$ and $A'$ (line 13 and 14, Algorithm 2). This operation should not increase the overall delay. For picking the vertex with highest index from $N$, we can perform the selection in a similar way.}. Moreover, the predecessor check is called at most $k\cdot min\{(n-k),k\Delta\}$ times. Therefore, the overall delay is $O(k^2\cdot min\{(n-k),k\Delta\}\cdot min\{k,\Delta\})$. Compared with the delay of $Simple$, 

\end{proof}

The bottleneck for the delay of $RwP$ is the time spent on the predecessor check. The original $RwP$ algoithm\cite{6} and the modified $RwP$ algorithm\cite{4} require $O(k(\Delta+\log k)min\{(n-k),k\Delta\})$ and $O(k^2\Delta)$ time for executing the predecessor check respectively. However, the predecessor check (lines 14-16) proposed in this work takes only $O(k\cdot min\{k,\Delta\})$ time. Hence, our proposed $IRwP$ algorithm has a better delay than both the original $RwP$ algorithm and the modified $RwP$ algorithm. Both the $IRwP$ algorithm and the $Simple$ algorithm \cite{4} require linear space. Comparing the delay of $IRwP$ and the delay of $Simple$, it can be seen that $Simple$ has a better delay in case of $1<k<n-1$.

\section{Conclusion}

In this work, we proposed a new neighborhood operator for constructing the supergraph $\mathcal{G}$ of the connected induced subgraphs of order $k$, and proved that the supergraph $\mathcal{G}$ constructed by the neighborhood operator is strongly connected. From a theoretical point of view, we improved upon the current best delay bound of algorithms for enumerating connected induced subgraphs in the case of large cardinality. 

\section{Acknowledgement}
The authors would like to thank the various grants from the National Natural Science Foundation of China (No. 61404069), Scientific Research Project of Colleges and Universities in Guangdong Province (No. 2021ZDZX1027), Guangdong Basic and Applied Basic Research Foundation (2022A1515110712), and STU Scientific Research Foundation for Talents (No. NTF20016 and No. NTF20017).



\begin{thebibliography}{1}

\bibitem{1}
Wernicke S. Efficient detection of network motifs. IEEE/ACM Transactions on Computational Biology and Bioinformatics, 2006, 3:347-359.
\bibitem{2}
Elbassuoni S, Blanco R. Keyword search over RDF graphs, in: Proceedings of the 20th ACM Conference on Information and Knowledge Management, (CIKM 2011), ACM. 2011, pp. 237-242.
\bibitem{3}
Karakashian S, Choueiry B.Y, and Hartke S.G. An algorithm for generating all connected subgraphs with k vertices of a graph. 2013. Technical Report. University of Nebraska-Lincoln. 2013.
\bibitem{4}
Komusiewicz C , Sommer F. Enumerating connected induced subgraphs: Improved delay and experimental comparison. Discrete Applied Mathematics, 2021, 303: 262-282.
\bibitem{5}
Alokshiya M, Salem S, Abed F. A linear delay algorithm for enumerating all connected induced subgraphs. BMC Bioinformatics, 2019, 20-S(319):1-11.
\bibitem{9}
R. Uehara. The number of connected components in graphs and its applications. Manuscript. URL: http://citeseerx.ist.psu.edu/viewdoc/summary?doi=10.1.1.44.5362.
\bibitem{6}
Elbassioni K. M. A polynomial delay algorithm for generating connected induced subgraphs of a given cardinality. Journal of Graph Algorithms and Applications, 2015, 19:273-280.
\bibitem{7}
Avis D and Fukuda K. Reverse search for enumeration. Discrete Applied Mathematics, 1996, 65(1-3):21-46.
\bibitem{8}
Tarjan R. Depth-first search and linear graph algorithms. SIAM journal on computing,1972, 1(2):146-160.
\bibitem{10}
Bollob\'{a}s, B., 2006. The Art of Mathematics - Coffee Time in Memphis. Cambridge University Press.
\end{thebibliography}

\end{document}